\documentclass[a4paper,11pt]{article}

\usepackage{fullpage}
\usepackage{soul}
\usepackage[utf8]{inputenc}
\usepackage[small]{caption}
\usepackage{xcolor}
\usepackage{amsmath}
\usepackage{booktabs}
\usepackage{algorithm}
\usepackage{enumitem}
\usepackage{framed}
\usepackage{tikz}
\usepackage{url}

\usepackage{amsthm}
\usepackage{amssymb}
\usepackage{afterpage}

\newtheorem{theorem}{Theorem}[section]

\newtheorem{proposition}[theorem]{Proposition}
\newtheorem{lemma}[theorem]{Lemma}

\theoremstyle{definition}
\newtheorem{definition}[theorem]{Definition}

\usepackage{natbib}

\allowdisplaybreaks

\usepackage{authblk}

\title{\bf Fairly Allocating Many Goods with Few Queries}

\author[1]{Hoon Oh} 
\author[2]{Ariel D. Procaccia}
\author[3]{Warut Suksompong}

\affil[1]{Computer Science Department, Carnegie Mellon University, USA}
\affil[2]{School of Engineering and Applied Sciences, Harvard University, USA}
\affil[3]{School of Computing, National University of Singapore, Singapore}

\date{}

\begin{document}

\maketitle

\begin{abstract}
We investigate the query complexity of the fair allocation of indivisible goods. For two agents with arbitrary monotonic utilities, we design an algorithm that computes an allocation satisfying envy-freeness up to one good (EF1), a relaxation of envy-freeness, using a logarithmic number of queries. We show that the logarithmic query complexity bound also holds for three agents with additive utilities, and that a polylogarithmic bound holds for three agents with monotonic utilities. These results suggest that it is possible to fairly allocate goods in practice even when the number of goods is extremely large. By contrast, we prove that computing an allocation satisfying envy-freeness and another of its relaxations, envy-freeness up to any good (EFX), requires a linear number of queries even when there are only two agents with identical additive utilities.
\end{abstract}

\section{Introduction}

Fair division is the study of how to allocate resources among interested agents in such a way that all agents find the resulting allocation to be fair~\citep{BT96,Moul03}. One of the field's paradigmatic applications is the allocation of \emph{indivisible} goods; this task typically arises in inheritance cases, when, say, an art or jewelry collection is divided between several heirs. Indeed, dividing goods is one of five applications offered by \emph{Spliddit}~\citep{GP14}, a not-for-profit fair division website; since its launch in November 2014, the website has served more than 130,000 users, and, in particular, has solved thousands of goods-division instances submitted by users.

While \citet{Stein48} was the first to study fairness from a mathematical point of view, the history of fair division actually goes back much further: A simple fair division mechanism called the \emph{cut-and-choose protocol} is mentioned in the Book of Genesis. After a dispute between Abraham and Lot, Abraham suggests that the two go their separate ways. He divides the land into two parts that---here we are perhaps using artistic license---he likes equally, and lets Lot choose the part that he prefers. The cut-and-choose protocol ensures that the resulting allocation satisfy an important fairness property called \emph{envy-freeness}---each of Abraham and Lot finds his part to be worth at least as much as the other person's part. Even though envy-freeness can always be satisfied when the allocated resources are divisible \citep{Strom80}, this is not the case when we deal with indivisible resources. With two agents and a single indivisible good, we already see that one of the agents will not receive the good and will therefore envy the other agent. 

Consequently, various relaxations of envy-freeness have been considered, the most prominent one being \emph{envy-freeness up to one good (EF1)}. This means that some agent may envy another agent under the given allocation, but that envy can be eliminated by removing a single good from the latter agent's bundle. \citet{LMMS04} showed that EF1 can be guaranteed even when the agents have arbitrary monotonic utilities. They achieved this by using an algorithm that we will refer to as the \emph{envy cycle elimination algorithm}, which runs in time $O(n^3m)$, where $n$ is the number of agents and $m$ the number of goods. 

\renewcommand{\arraystretch}{1.5}
\begin{table*}[t]
\begin{center}
    \begin{tabular}{ | c | c | c | }
    \hline
     Notion & Monotonic utilities & Additive utilities \\ \hline \hline
		EF & $\geq\binom{m}{m/2}$ (Prop. \ref{prop:EFmonotonic}) & $\Theta(m)$ (Th. \ref{thm:EFadditive})   \\ \hline
		EFX & $\Omega\left(\frac{1}{m}\binom{m}{(m-1)/2}\right)$ \citep{PlautRo20}  & $\Theta(m)$ (Th. \ref{thm:EFXadditive})  \\ \hline
		EF1 & $\Theta(\log m)$ (Th. \ref{thm:EF1monotonic}, Prop. \ref{prop:EF1binary})  & $\Theta(\log m)$ (Th. \ref{thm:EF1monotonic}, Prop. \ref{prop:EF1binary})  \\ \hline
    \end{tabular}
    \vspace{5mm}
    \caption{Query complexity in the setting with two agents. All lower bounds hold even when the two agents have identical utilities.}
    \label{table:summary}
\end{center}
\end{table*}

Given that an EF1 allocation always exists and can be found efficiently at this level of generality, a natural question to ask is how much we need to know about the agents' utilities to compute such an allocation. This issue is crucial for combinatorial utilities, since merely writing down a complete utility function might already take exponential time. But the question is equally important for additive utilities; while expressing such a utility function only takes linear time, this may already be prohibitive if the number of goods is very large. In fact, the goods application on Spliddit elicits additive utilities and computes an EF1 allocation~\citep{CKMP+16}; the largest instance that was encountered involved ten siblings and roughly 1400 goods. In this case, the siblings actually prepared a spreadsheet with their value for each of the goods!
As we will see, perhaps surprisingly, in several cases there do exist algorithms that compute EF1 allocations using only a logarithmic number of queries, despite the fact that the size of the input can be linear or even exponential.\footnote{We also remark here that our algorithms with logarithmic query complexity allow agents to specify an arbitrary (common) ordering of the goods, and only request their utilities of subsets that are composed of a (small) constant number of contiguous blocks with respect to this ordering. In particular, the algorithms do not request utilities of ``random'' subsets like $\{1,3,6,7,10,13,15,\dots,1397,1400\}$, which can be cumbersome for the agents to determine.}

\subsection{Our Results}

We allow algorithms to elicit the utilities of agents via a standard interface, \emph{value queries}, which ask an agent for her value for a given subset of goods. We assume that algorithms are deterministic and measure their complexity in terms of the worst-case number of queries that they require.

In Section~\ref{sec:two} we consider the case of two agents. We show that it is possible to compute an EF1 allocation for agents with arbitrary monotonic utilities using a logarithmic number of queries (Theorem~\ref{thm:EF1monotonic}). This is asymptotically tight, even for two agents with identical and very simple binary utilities (Proposition~\ref{prop:EF1binary}). We then turn to envy-freeness and establish that determining whether an envy-free allocation exists takes an exponential number of queries for agents with identical monotonic utilities (Proposition~\ref{prop:EFmonotonic}) and a linear number of queries for agents with identical additive utilities (Theorem~\ref{thm:EFadditive}); our latter bound is also exactly tight. We end our investigation of the two-agent case by considering another relaxation of envy-freeness called \emph{envy-freeness up to any good (EFX)}, a stronger notion than EF1. We show that computing an EFX allocation already takes a linear number of queries for agents with identical additive utilities (Theorem~\ref{thm:EFXadditive}). This complements a recent result of \citet{PlautRo20}, who showed that while an EFX allocation always exists for two agents with arbitrary monotonic utilities, computing one such allocation already requires an exponential number of queries in the worst case, even when the utilities of the agents are identical. Taken together, these results suggest that, when the number of goods is large, EF1 is the `right' notion of fairness, whereas EFX is too demanding. 
The results of Section~\ref{sec:two} are summarized in Table~\ref{table:summary}.

In Section~\ref{sec:three} we address the case of three agents. Our main result is an algorithm that computes an EF1 allocation for three agents with additive utilities using a logarithmic number of queries (Theorem~\ref{thm:threeadditivealgo}). Our algorithm adapts the Selfridge-Conway procedure, a classical \emph{cake-cutting} protocol for computing an envy-free allocation of a heterogeneous \emph{divisible} good, to the setting of indivisible goods. In particular, as a building block we use an algorithm that, for three agents with \emph{identical} additive utilities, computes an EF1 allocation satisfying the extra property that any three predetermined goods belong to three different bundles (Lemma~\ref{LEM:SEPARATE}). Moreover, we show that by adapting a recent algorithm of \citet{BCFI+19}, it is possible to compute an EF1 allocation for three agents with arbitrary monotonic utilities using a polylogarithmic number of queries (Theorem~\ref{thm:threemonotonic}).

Finally, in Section~\ref{sec:many} we consider the setting where there can be any number of agents. We show that the envy cycle elimination algorithm of \citet{LMMS04} can be implemented using a relatively modest number of queries (Theorem~\ref{thm:manyupper-envygraph}). In addition, we propose algorithms that use fewer queries under stronger assumptions on agents' utilities, such as when the agents have identical monotonic utilities (Theorem~\ref{thm:manyupper-identical}). 
To complement these positive results, we conclude by presenting a lower bound on the number of queries needed to compute an EF1 allocation (Theorem~\ref{thm:manylower}).

\subsection{Related Work}

The paper that is most closely related to ours is the one mentioned above, by \citet{PlautRo20}. Using an interesting reduction from the local search problem on a class of graphs known as Kneser graphs, they show that the problem of finding an EFX allocation requires an exponential number of queries, even for two agents with identical utilities. They also examine when EFX can be achieved in conjunction with other properties such as Pareto optimality, and establish the existence of allocations satisfying an approximate version of EFX for agents with subadditive utilities.
In a follow-up paper, \citet{PlautRo20-2} explore communication complexity in the fair allocation of discrete items.

A bit further afield, query complexity has long been a topic of interest in computational fair division, albeit in the context of \emph{divisible} goods, also known as cake cutting \citep{Pro13}. The standard query model for cake cutting is due to \citet{RW98}, and allows two types of operations: evaluate (which is similar to our value queries) and cut. In this model, the query complexity of achieving fair cake allocations, under various notions of fairness, is well-studied \citep{Pro09,EP06,DQS12,KLP13,AM16b,BN17,PW17,ElkindSeSu21}.
Note that in cake cutting, it is commonly assumed that the cake forms an interval and the queries are made with respect to the interval. On the other hand, in our indivisible goods setting the goods do not inherently lie on a line, so there is no equivalence of the cut query in our model.
Nevertheless, a recurring technique that we will use in order to achieve (poly)logarithmic query complexity is to arrange the goods in a line and perform binary search to determine a cut point of interest.\footnote{The idea of arranging goods in a line and adapting cake-cutting procedures was also used in the work of \citet{BCFI+19}, which appeared after we published an initial version of our paper.}

The envy-freeness relaxations EF1 and EFX have received considerable attention from the computer science community in the past few years, with a number of papers studying various aspects of these notions \citep{ABM18,BarmanKrVa18,BB18,BeiLuMa19,BeiIgLu21,BCFI+19,CaragiannisGrHu19,CKMP+16,KyropoulouSuVo20}.
In particular, \citet{Suksompong20} provided a partial explanation for the large gap in query complexity between EF1 and EFX by showing that for two agents with arbitrary monotonic utilities, the number of EF1 allocations is always exponential in the number of goods, while there can be as few as two EFX allocations regardless of the number of goods.
In addition to these relaxations, another well-studied fairness criterion for indivisible goods is the \emph{maximin share} \citep{Bud11,GHSS+17,KPW18,BarmanKr20}.

\section{Preliminaries}

There is a set $G=\{g_1,g_2,\dots,g_m\}$ of goods and a set $A=\{a_1,a_2,\dots,a_n\}$ of agents. A \emph{bundle} is a subset of $G$. Each agent $a_i$ has a nonnegative utility $u_i(G')$ for each $G'\subseteq G$. We sometimes abuse notation and write $u_i(g)$ for $u_i(\{g\})$. 

A utility function is said to be \emph{monotonic} if $u_i(G_1)\leq u_i(G_2)$ for any $i$ and any $G_1\subseteq G_2\subseteq G$. It is said to be \emph{additive} if $u_i(G')=\sum_{g\in G'}u_i(g)$ for any $G'\subseteq G$, and \emph{binary} if it is additive and $u_i(g)=0$ or $1$ for each $g\in G$.
While additivity is significantly more restrictive than monotonicity, many papers in fair division assume that agents' utilities are additive \citep{BL16,AMNS15,KPW18,CKMP+16}. This assumption is also made by Spliddit's app for dividing goods \citep{CKMP+16}, as, in practice, additive utilities hit a sweet spot between expressiveness and ease of elicitation. We assume throughout the paper that agents have monotonic utilities\footnote{Without this assumption, neither of the relaxations of envy-freeness that we consider can always be satisfied even when there are two agents. Indeed, for EF1, suppose that there is a single good and both agents prefer not getting the good to getting it. Then the agent who gets the good will always be envious, and there is no possibility of removing a good from the other agent's bundle since it is empty. For EFX, assume that there are two goods, and the two agents share a common utility function $u$ with $u(\emptyset)=1$, $u(\{g_1\})=u(\{g_1,g_2\})=0$, and $u(\{g_2\})=2$. Then if one agent gets no good, she is envious when we remove $g_1$ from the other agent's bundle. Else, if both agents get a good, the agent who gets $g_1$ is envious when we remove $g_2$ from the other agent's bundle.} 
and that, without loss of generality, $u_i(\emptyset)=0$ for all $i$. 

An \emph{allocation} is a partition of $G$ into $n$ bundles $(G_1,G_2,\dots,G_n)$, where bundle $i$ is allocated to agent $i$. If the goods lie on a line, for each good $g$ we denote by $L_g$ and $R_g$ the set of goods to the left and right of $g$, respectively. A \emph{contiguous} allocation is an allocation in which every bundle forms a contiguous block on the line.

We now define the fairness notions that we consider.
\begin{definition}
An allocation $(G_1,G_2,\dots,G_n)$ is said to be 
\begin{itemize}
\item \emph{envy-free} if $u_i(G_i)\geq u_i(G_j)$ for any $i,j$. 
\item \emph{envy-free up to any good (EFX)} if for any $i,j$ and any good $g\in G_j$, $u_i(G_i)\geq u_i(G_j\backslash\{g\})$.
\item \emph{envy-free up to one good (EF1)} if for any $i,j$ such that $u_i(G_i)<u_i(G_j)$, there exists a good $g\in G_j$ such that $u_i(G_i)\geq u_i(G_j\backslash\{g\})$.\footnote{The clause ``such that $u_i(G_i)<u_i(G_j)$'' is necessary for the case where $G_j=\emptyset$.}
\end{itemize}
\end{definition}
It is easy to see that envy-freeness is stronger than EFX, which is in turn stronger than EF1. Envy-freeness is a classical and well-studied fairness notion that goes back to \citet{Fol67}. By contrast, its two relaxations are relatively new: EF1 was introduced by \citet{Bud11} and a related property was studied by \citet{LMMS04}, while EFX was only proposed recently by \citet{CKMP+16}. 

We will consider algorithms that compute fair allocations according to these fairness notions. In order to discover the agents' utilities, an algorithm is allowed to issue \emph{value queries}. In each query, the algorithm chooses an agent $a_i$ and a subset $G'\subseteq G$, and finds out the value of $u_i(G')$. We assume that the algorithm is deterministic, and allow it to be \emph{adaptive}, i.e., the algorithm can determine its next query based on its past queries and the corresponding answers.

\section{Two Agents}
\label{sec:two}

In this section, we consider the setting with two agents. We organize our results based on fairness notion: EF1, envy-freeness, and EFX. 

\subsection{EF1}

We begin by describing an algorithm that computes an EF1 allocation for two agents with arbitrary monotonic utilities. The algorithm is similar to the cut-and-choose protocol for cake cutting: the first agent partitions the goods into two bundles with the property that she would be satisfied with either bundle, and the second agent chooses the bundle that she prefers. In order to minimize the number of queries, we arrange the goods on a line and use binary search to determine the cut point of the first agent.
Recall that for each good $g$, we denote by $L_g$ and $R_g$ the set of goods to the left and right of $g$, respectively.

\begin{framed}
\noindent
\textbf{Algorithm~1} (for two agents with monotonic utilities) \\

\noindent
\emph{Step~1:} Arrange the goods on a line in arbitrary order. Find the rightmost good $g$ such that  $u_1(L_g)\leq u_1(R_g\cup\{g\})$. \\

\noindent
\emph{Step~2:} If $u_1(L_g)\leq u_1(R_g)$, consider the partition $(L_g\cup\{g\}, R_g)$; else, consider the partition $(L_g, R_g\cup\{g\})$. Give $a_2$ the bundle from the partition that she prefers, and $a_1$ the remaining bundle.
\end{framed}

We claim that the algorithm computes an EF1 allocation using a logarithmic number of queries.

\begin{theorem}
\label{thm:EF1monotonic}
For two agents with arbitrary monotonic utilities, Algorithm~1 computes an EF1 allocation. Moreover, the algorithm can be implemented to use $O(\log m)$ queries in the worst case.
\end{theorem}

\begin{proof}
We first show that the algorithm computes an EF1 allocation. Since $a_2$ gets the bundle that she prefers, she does not envy $a_1$. 

To reason about $a_1$'s envy, assume first that $u_1(L_g)\leq u_1(R_g)$. It holds that $u_1(L_g\cup\{g\})\geq u_1(R_g)$: This is clearly true if $g$ is the rightmost good on the line, and otherwise it follows from the definition of $g$ that $u_1(R_g)<u_1(L_g\cup\{g\})$. Therefore, if $a_1$ receives $L_g\cup\{g\}$, she is not envious at all. And if she receives $R_g$, it holds that 
$$u_1(R_g)\geq u_1(L_g)= u_1((L_g\cup\{g\})\setminus \{g\}),$$ 
so EF1 is satisfied. 

The second case is where $u_1(L_g)>u_1(R_g)$. If $a_1$ gets $R_g\cup\{g\}$ then she is not envious, since, by the definition of $g$, $u_1(R_g\cup\{g\})\geq u_1(L_g)$. If she gets $L_g$ instead, then EF1 holds, because 
$$u_1(L_g)>u_1(R_g)=u_1((R_g\cup\{g\})\setminus \{g\}).$$

Next, we show that the algorithm can be implemented to use $O(\log m)$ queries. By monotonicity, Step~1 can be done by binary search using $O(\log m)$ queries. In Step~2, we use two queries to compare $u_1(L_g)$ and $u_1(R_g)$, and two more queries to compare $a_2$'s utility for the two bundles in the partition. Hence the total number of queries is $O(\log m)$.
\end{proof}

The following proposition shows that the bound $O(\log m)$ in Theorem~\ref{thm:EF1monotonic} is tight. 

\begin{proposition}
\label{prop:EF1binary}
Any deterministic algorithm that computes an EF1 allocation for two agents with identical binary utilities uses $\Omega(\log m)$ queries in the worst case, even when each agent values only two goods.
\end{proposition}

Since Proposition~\ref{prop:EF1binary} will later be generalized by Theorem~\ref{thm:manylower}, we do not present its proof here.

\subsection{Envy-freeness}

Next, we turn our attention to envy-freeness. Unlike the case of EF1, allocations that satisfy envy-freeness are not guaranteed to exist.  We show that for two agents with identical monotonic utilities, even an algorithm that only decides whether an envy-free allocation exists already needs to make an exponential number of queries in the worst case. A similar argument holds for algorithms that compute an envy-free allocation whenever one exists. 

\begin{proposition}
\label{prop:EFmonotonic}
Assume that $m$ is even. Any deterministic algorithm that determines whether an envy-free allocation exists for two agents with identical monotonic utilities uses at least $\binom{m}{m/2}$ queries in the worst case.
\end{proposition}

\begin{proof}
Assume that the common utility function $u$ is such that $u(G_1)<u(G_2)$ whenever $|G_1|<|G_2|$. With this utility function, any envy-free allocation must give an equal number of goods to both agents. Since the values of subsets of size $m/2$ can be arbitrary, the algorithm must, in the worst case, query all such subsets in order to determine whether there exists a set $G'$ such that $u(G')=u(G\backslash G')$. Hence the algorithm needs to query $\binom{m}{m/2}$ subsets in the worst case.
\end{proof}

Even though an algorithm that decides whether an envy-free allocation exists needs to make an exponential number of queries for agents with monotonic utilities, when we restrict our attention to agents with additive utilities, the exponential lower bound no longer holds since the algorithm can query the value of both agents for every good and find out the full utility functions. It is conceivable that there are algorithms that do even better asymptotically, e.g., use a logarithmic number of queries. However, we show that this is not the case: a linear number of queries is necessary, even when the two agents have identical utilities.\footnote{When the two agents have identical utilities, the problem of determining whether an envy-free allocation exists is clearly equivalent to the well-known \textsc{Partition} problem. However, the NP-hardness of \textsc{Partition} does not carry over to our setting, as we are only interested in the query complexity.} In fact, we leverage linear-algebraic techniques to establish that at least $m$ queries are needed in this case---to the best of our knowledge, this is the first time that such techniques have been used to show lower bounds in fair division. This bound is tight for two identical agents since the algorithm can find out the common utility function by querying the value of each of the $m$ goods.

\begin{theorem}
\label{thm:EFadditive}
Assume that $m$ is even. Any deterministic algorithm that decides whether an envy-free allocation exists for two agents with identical additive utilities uses at least $m$ queries in the worst case.
\end{theorem}

\begin{proof}
For ease of notation, let $x_i=u(g_i)$ for $i=1,2,\dots,m$, where $u$ is the common utility function. Note that an envy-free allocation exists if and only if the goods can be partitioned into two sets of equal value. Consider an algorithm that always uses at most $m-1$ queries. Assume without loss of generality that the algorithm always uses exactly $m-1$ queries; whenever it uses fewer than $m-1$ queries, we add arbitrary queries for the algorithm. The key idea is that for each query, if the queried subset has size $s$, we will give an answer close to $s$ in such a way that after all $m-1$ queries, it is still possible that there exists an envy-free allocation, but also that there does not exist one. This will allow us to obtain the desired conclusion.

For $i=1,2,\dots,m-1$, let $\mathbf{v}_i$ be a vector of length $m$ where the $j$th component is 1 if good $g_j$ is included in the $i$th query, and 0 otherwise. Therefore the $i$th query asks for the value $\mathbf{v}_i\cdot\mathbf{x}=\sum_{j=1}^m v_{i,j}x_j$. Furthermore, let $W$ be the set of all vectors of length $m$ all of whose components are $\pm 1$ (so $|W|=2^m$), and let $W'\subset W$ be the set of vectors with an equal number of $-1$ and 1. Note that an envy-free allocation exists exactly when $\mathbf{w}\cdot\mathbf{x}=0$ for some $\mathbf{w}\in W$.

When we receive the $i$th query, if $\mathbf{v}_i\in\text{span}(\mathbf{v}_1,\mathbf{v}_2,\dots,\mathbf{v}_{i-1})$, our answer is already determined by previous answers. Hence, we may assume without loss of generality that $\mathbf{v}_i\not\in\text{span}(\mathbf{v}_1,\mathbf{v}_2,\dots,\mathbf{v}_{i-1})$ for all $i$. For each $\mathbf{w}\in W$ such that $\mathbf{w}\in\text{span}(\mathbf{v}_1,\mathbf{v}_2,\dots,\mathbf{v}_i)\backslash\text{span}(\mathbf{v}_1,\mathbf{v}_2,\dots,\mathbf{v}_{i-1})$, there exists a unique answer that would force $\mathbf{w}\cdot\mathbf{x}=0$. We avoid all such (finite number of) answers. After query number $m-1$ we have a subspace $\mathcal{V}=\text{span}(\mathbf{v}_1,\mathbf{v}_2,\dots,\mathbf{v}_{m-1})$ of dimension $m-1$ such that we know the value $\mathbf{v}\cdot\mathbf{x}$ if and only if $\mathbf{v}$ is in the subspace. 

Next, let $\mathcal{W}'=\text{span}(W')$. Clearly, all vectors in $\mathcal{W}'$ are orthogonal to the vector $(1,1,\dots,1)$. We claim that $\mathcal{W}'$ in fact has dimension $m-1$, and therefore consists of \emph{all} vectors orthogonal to $(1,1,\dots,1)$. To see this, take two distinct vectors in $W'$ that differ only in the first and $i$th component for some $i=2,3,\dots,m$. The difference vector, which consists of a 2 in the first position, a $-2$ in the $i$th position, and 0 elsewhere belongs to $\mathcal{W}'$. It is clear that no nontrivial linear combination of these difference vectors can produce the all-zero vector, meaning that the $m-1$ vectors are linearly independent, and thus $\mathcal{W}'$ indeed has dimension $m-1$. 

Now, since any vector $\mathbf{v}_i$ is not orthogonal to $(1,1,\dots,1)$, we have $\mathcal{V}\neq \mathcal{W}'$, and so there exists $\mathbf{w}'\in W'$ such that $\mathbf{w}'\not\in\mathcal{V}$. (If this were not the case, we would have $\mathcal{W}'\subseteq\mathcal{V}$, and then the two subspaces would be equal because they are of the same dimension.) Since $\mathcal{V}$ is of dimension $m-1$ and $\mathbf{w}'\not\in\mathcal{V}$, setting the value of $\mathbf{w}'\cdot\mathbf{x}$ will, in combination with the constraints resulting from our answers to the $m-1$ queries, uniquely determine $\mathbf{x}$. If we set $\mathbf{w}'\cdot\mathbf{x}=0$, an envy-free allocation exists. On the other hand, if we set $\mathbf{w}'\cdot\mathbf{x}$ so that $\mathbf{w}\cdot\mathbf{x}\neq 0$ for all $\mathbf{w}\in W$, an envy-free allocation does not exist. This choice of value for $\mathbf{w}'\cdot\mathbf{x}$ is available because for each $\mathbf{w}\in W$, only one value of $\mathbf{w}'\cdot\mathbf{x}$ forces $\mathbf{w}\cdot\mathbf{x}=0$.

It remains to show that we can give the answers in such a way that after setting the value of $\mathbf{w}'\cdot\mathbf{x}$, all components of the unique solution for $\mathbf{x}$ are nonnegative. We choose $\delta>0$ such that for any vector $\mathbf{z}$ with $|z_i|<\delta$ for all $i=1,2,\dots,m$, and any $m\times m$ invertible matrix $M$ all of whose entries are $-1$, 0, or 1, the unique solution $\mathbf{y}$ to $M\mathbf{y}=\mathbf{z}$ has $|y_i|<1$ for all $i$. 
To see that this choice of $\delta$ can be made, note that we may find $\delta_M$ for each such matrix $M$ and take $\delta$ to be the minimum among the (finite number of) values $\delta_M$.
Each $M$ can be viewed as a linear transformation that takes the all-zero vector to itself, so the inverse transformation takes the ball of radius $1$ centered at the origin to within the ball of some radius $r_M$ centered at the origin. We can then take $\delta_M := 1/r_M$.

For each query on a subset of size $k$, we give an answer in the range $(k-\delta,k+\delta)$. Moreover, we choose the value of $\mathbf{w}'\cdot\mathbf{x}$ to be in the range $(-\delta,\delta)$. 
Note that these choices are always possible since for each answer, there are only a finite number of forbidden choices.
Write $y_i=x_i-1$ for all $i$, where $\mathbf{x}$ is the unique solution according to our choices. Our answers to the queries ensure that the values of $\mathbf{v}_i\cdot\mathbf{y}$ for $i=1,2,\dots,m-1$ belong to the range $(-\delta,\delta)$, and our choice of $\mathbf{w}'\cdot\mathbf{x}$ ensures that $\mathbf{w}'\cdot\mathbf{y}$ also belong to this range.
Take $M$ to be the $m\times m$ matrix with $\mathbf{v}_1,\mathbf{v}_2,\dots,\mathbf{v}_{m-1}$ and $\mathbf{w}'$ as its rows, and $\mathbf{z}$ to be a vector of length $m$ with $\mathbf{v}_1\cdot\mathbf{y},\mathbf{v}_2\cdot\mathbf{y},\dots,\mathbf{v}_{m-1}\cdot\mathbf{y}$ and $\mathbf{w}'\cdot\mathbf{y}$ as its elements.
Since all elements of $\mathbf{v}_i$ and $\mathbf{w}'$ belong to the set $\{-1,0,1\}$, all entries of $M$ also belong to this set, and hence our definition of $\delta$ implies that $|y_i|<1$ for all $i$. It follows that $x_i>0$ for all $i$, as desired.
\end{proof}

For two agents with additive utilities, envy-freeness is equivalent to another well-known fairness notion called proportionality, which requires that each agent receive at least half of her value for the whole set of goods. Thus, the lower bound in Theorem~\ref{thm:EFadditive} also holds for two agents with identical additive utilities with respect to proportionality.

\subsection{EFX}

We end this section by considering EFX. For two agents with monotonic utilities, \citet{PlautRo20} showed that an EFX allocation is guaranteed to exist, but computing it takes an exponential number of queries in the worst case. If the agents have additive utilities, however, the algorithm can already find out the full utility functions using only a linear number of queries. Our next result shows that a linear number of queries is, in fact, needed for computing an EFX allocation.

\begin{theorem}
\label{thm:EFXadditive}
Assume that $m$ is odd. Any deterministic algorithm that computes an EFX allocation for two agents with identical additive utilities uses at least $(m-1)/2$ queries in the worst case.
\end{theorem}

\begin{proof}
Let $m=2k+1$, and consider an algorithm that always uses at most $k-1$ queries. Whenever the algorithm queries a subset of size $s$, we answer that the subset has value $s$. Suppose that upon termination the algorithm outputs the allocation $(G_1,G_2)$, where we assume without loss of generality that $|G_1|\leq k$. If $|G_1|<k$, then we assign to every good a value of 1; this is clearly consistent with our answers, but the allocation is not EFX. 

Assume, therefore, that $|G_1|=k$. We assign to every good in $G_1$ a value of 1; this is again consistent with our answers to queries on subsets of $G_1$. Now, our answers give rise to at most $k-1$ linear constraints on the values of the goods in $G_2$. Add the constraint that the sum of these $k+1$ values is $k+1$. The allocation $(G_1,G_2)$ is EFX if and only if every good in $G_2$ has value 1. Since all of our constraints are satisfied when every good has value 1, our set of constraints is satisfiable. There are $k+1$ variables and at most $k$ constraints, so if we compute the reduced row-echelon form of the constraint matrix, we can find a nonempty set of free variables. The remaining variables (i.e., the leading variables) can be written as a linear combination of these free variables. If all free variables are set to 1, all leading variables must also be 1. Hence, we can perturb one of the free variables by a small amount so that all of the leading variables are still nonnegative. This yields a utility function that is consistent with our answers but according to which the allocation returned by the algorithm is not EFX.
\end{proof}

Note that Theorem~\ref{thm:EFXadditive} is incomparable with Theorem~\ref{thm:EFadditive}, even though EFX is a relaxation of envy-freeness, because the former result deals with a \emph{search} problem (finding an EFX allocation knowing that one always exists), whereas the latter deals with a \emph{decision} problem (deciding whether an EF allocation exists at all).

\section{Three Agents}
\label{sec:three}

In this section, we study the setting with three agents who are endowed with additive utilities. Our main result is an algorithm that finds an EF1 allocation using $O(\log m)$ queries, but we first need to develop some machinery for the case where the agents have identical utilities. 
Indeed, our algorithms for computing an EF1 allocation in the case of identical additive utilities (Algorithm~2 and Lemma~\ref{LEM:SEPARATE}) will be important building blocks in the algorithm for general additive utilities.

\subsection{Identical Additive Utilities}

While the case of identical additive utilities might seem trivial at first glance, as we will see, there are already several interesting statements that we can make about the setting, so it may be of independent interest. We begin by establishing some properties of a particular partition of goods on a line into two contiguous blocks.

\begin{lemma}
\label{lem:cutandchoose}
Assume that the goods lie on a line. Suppose that an agent with an additive utility function $u$ chooses the partition of the goods into two contiguous blocks that minimizes the difference between the values of the two blocks. (If there are goods of value 0 next to the cut point, move the cut point until these goods belong to the block of lower value.) Let $L$ be the left block and $g_l$ the rightmost good of the block. Similarly, let $R$ be the right block and $g_r$ the leftmost good of the block. Then:
\begin{enumerate}
\item We have that $\min\{u(G)/2,u(L)\}\geq u(R\backslash\{g_r\})$ and $\min\{u(G)/2,u(R)\}\geq u(L\backslash\{g_l\})$.
\item The partition can be computed using $O(\log m)$ queries in the worst case.
\end{enumerate}
\end{lemma}

\begin{proof}
We prove the two parts in turn.
\begin{enumerate}
\item By symmetry it suffices to prove that $\min\{u(G)/2,u(L)\}\geq u(R\backslash\{g_r\})$. Suppose that $u(L)<u(R\backslash\{g_r\})$. We have $u(g_r)>0$ since otherwise we would have moved the cut point to the right. Then $u(L\cup\{g_r\})>u(L)$ and $u(R\backslash\{g_r\})>u(L)$; hence $(L\cup\{g_r\},R\backslash\{g_r\})$ is a more equal partition than $(L,R)$, a contradiction. If $u(R\backslash\{g_r\})>u(G)/2$ then we also have that $u(R\backslash\{g_r\})>u(L)$, and the same argument yields a contradiction.
\item  We use binary search to find the leftmost good $g$ such that $u(L_g\cup\{g\})\geq u(G)/2$. The left block of the partition will be either $L_g$ or $L_g\cup\{g\}$. Indeed, if the left block is smaller than $L_g$ then having $L_g$ as the left block yields a more equal partition, while if it is larger than $L_g\cup\{g\}$ then having $L_g\cup\{g\}$ as the left block yields a more equal partition. Moreover, with this choice of partition, any good of value 0 next to the cut point will already belong to the block of lower value. The total number of queries required for the binary search is $O(\log m)$.
\end{enumerate}
The proof is complete.
\end{proof}

Next, we present an algorithm that computes a contiguous EF1 allocation for three agents with identical utilities using a logarithmic number of queries. Not only will the contiguity condition be useful later in our algorithm for three agents with arbitrary utilities, but in certain applications it may also be desirable to produce a contiguous allocation. For example, if the goods are office space, it is conceivable that each research group wishes to have a consecutive block of offices in order to facilitate collaboration within the group. While contiguous fair allocations of indivisible goods have recently been studied \citep{BCEI+17,Suk17}, to the best of our knowledge even the \emph{existence} of a contiguous EF1 allocation for three agents with identical utilities has not been established before, let alone an algorithm that computes such an allocation using a small number of queries. Hence our result may be of interest even if one is not concerned with the number of queries made. In Section~\ref{sec:many}, the existence result is generalized to any number of agents with identical \emph{monotonic} utilities (Lemma~\ref{lem:manycontiguous}).\footnote{After we published an initial version of our paper, \citet{BCFI+19} independently proved this generalization. In addition, they showed the existence of a contiguous EF1 allocation for up to four agents with arbitrary monotonic utilities.}

To demonstrate that the problem of establishing the existence of a contiguous EF1 allocation in this setting is not straightforward, we present a very natural approach that, perhaps surprisingly, does not work. We first pretend that the goods are divisible and find the two cut points that would divide the goods into three parts of exactly equal value. For each cut point, if the cut point falls between two (now indivisible) goods, we keep it; otherwise we round it either to the left or to the right. One might be tempted to claim that at least one of the resulting allocations must be EF1. Indeed, Lemma~\ref{lem:cutandchoose} implies that an analogous approach works for two agents. However, an example given in Appendix~\ref{app:ex} shows that the approach does not work for three agents, no matter how we round the cut points.

\begin{framed}
\noindent
\textbf{Algorithm~2} (for three agents with \emph{identical} additive utilities) \\

\noindent
\emph{Step~1:} Assume that the goods lie on a line, and denote by $u$ the common utility function of the three agents. Let $g_1$ be the leftmost good such that $u(L_{g_1}\cup \{g_1\})>u(G)/3$, and let $g_2$ be the rightmost good such that $u(R_{g_2}\cup \{g_2\})>u(G)/3$. (Possibly $g_1=g_2$.) Assume without loss of generality that $u(L_{g_1})\geq u(R_{g_2})$; the algorithm proceeds analogously in the opposite case.  \\

\noindent
\emph{Step~2:} If $L_{g_1}\neq\emptyset$, let $g_3$ be the leftmost good such that $u(L_{g_3}\cup \{g_3\})\geq u(R_{g_2})$. Set $A=L_{g_3}\cup\{g_3\}$. Else, set $A=\emptyset$. In both cases, set $C=R_{g_2}$ and $B=G\backslash(A\cup C)$. \\

\noindent
\emph{Step~3:} If $u(C)\geq u(B\backslash\{g_2\})$, return the allocation $(A,B,C)$. Else, set $C'=R_{g_2}\cup\{g_2\}$. Partition the remaining goods into two contiguous blocks according to Lemma~\ref{lem:cutandchoose}; denote by $A'$ the left block and $B'$ the right block. Return the allocation $(A',B',C')$.
\end{framed}

The following lemma establishes the claimed properties of Algorithm~2.

\begin{lemma}
\label{lem:threeidenticalalgo}
Assume that the goods lie on a line. For three agents with identical additive utilities, Algorithm~2 computes a contiguous EF1 allocation. Moreover, the algorithm can be implemented to use $O(\log m)$ queries in the worst case.
\end{lemma}

\begin{proof}
We first show that the algorithm computes an EF1 allocation. We assume like at the end of Step~1 that $u(L_{g_1})\geq u(R_{g_2})$ (the proof for the opposite case is analogous), and consider two cases.

\begin{itemize}
\item \emph{Case~1}: At the beginning of Step~3, $u(C)\geq u(B\backslash\{g_2\})$ (so we return the allocation $(A,B,C)$ from Step~2). We have $u(A)\geq u(C)\geq u(B\backslash\{g_2\})$, so $a_1$ and $a_3$ do not envy $a_2$ up to one good. Since $u(B)\geq u(G)/3\geq u(A)\geq u(C)$, $a_2$ does not envy $a_1$ or $a_3$, and $a_1$ does not envy $a_3$. Moreover, if $L_{g_1}\neq\emptyset$, then by definition of $g_3$ we have $u(A\backslash\{g_3\})\leq u(C)$, which means that $a_3$ does not envy $a_1$ up to one good. If $L_{g_1}=\emptyset$ then $A=C=\emptyset$, and again $a_3$ does not envy $a_1$ up to one good.

\item \emph{Case~2}: At the beginning of Step~3, $u(C)< u(B\backslash\{g_2\})$. We also have $u(C)\leq u(A)$. Since we partition $A\cup (B\backslash\{g_2\})$ into two blocks according to Lemma~\ref{lem:cutandchoose}, both blocks are also worth at least $u(C)=u(C'\backslash\{g_2\})$, meaning that $a_1$ and $a_2$ do not envy $a_3$ up to one good. We claim that $a_2$ and $a_3$ also do not envy $a_1$ up to one good; the claim for $a_1$ and $a_3$ towards $a_2$ can be shown similarly. Let $g$ be the rightmost good in $A'$. (If $A'=\emptyset$, the claim holds trivially.) It suffices to show that $u(B')\geq u(A'\backslash\{g\})$ and $u(C')\geq u(A'\backslash\{g\})$. By Lemma~\ref{lem:cutandchoose}, we have $u(B')\geq u(A'\backslash\{g\})$ and 
\begin{align*}
u(A'\backslash\{g\})&\leq u(A'\cup B')/2\\
&=\frac{u(G)-u(C')}{2} \\
&\leq u(G)/3\\
&\leq u(C'),
\end{align*}
Hence the allocation is EF1.
\end{itemize}

Next, we show that the algorithm can be implemented to use $O(\log m)$ queries. By monotonicity, both finding $g_1$ and $g_2$ in Step~1 and finding $g_3$ in Step~2 can be done by binary search using $O(\log m)$ queries. By Lemma~\ref{lem:cutandchoose}, the partition in Step~3 can be found using $O(\log m)$ queries. The remaining operations of the algorithm only require a constant number of queries. Hence the total number of queries is $O(\log m)$.
\end{proof}

A bonus of Algorithm~2 is that in the allocation produced by the algorithm, if some agent envies another agent, then the envy can be eliminated by removing not just some arbitrary good from the latter agent's bundle, but one of the goods at the end of the latter agent's block. In fact, we can also choose this good to be a good next to a cut point; this nails down a unique good for the agents getting the left or right block. The property can be deduced from the proof of Lemma~\ref{lem:threeidenticalalgo}.

Next, we leverage Algorithm~2 to show that for three agents with identical additive utilities, if we designate three goods in advance, it is possible to compute an EF1 allocation such that all three designated goods belong to different bundles. 

\begin{lemma}
\label{LEM:SEPARATE}
Let $g_1$, $g_2$, $g_3$ be three distinct goods. For three agents with identical additive utilities, there exists a deterministic algorithm that computes an EF1 allocation such that the three goods belong to three different bundles using $O(\log m)$ queries in the worst case.
\end{lemma}

Before we establish the lemma, we show an interesting property of Algorithm~2.

\begin{lemma}
\label{lem:middle}
Assume that the goods lie on a line, and let $g$ be a good such that $u(L_g)\geq u(G)/3$ and $u(R_g)\geq u(G)/3$. Then $g$ belongs to the middle bundle in the allocation returned by Algorithm~2.
\end{lemma}

\begin{proof}
By definition of $g_1$ in Algorithm~2, we have that $g$ is either $g_1$ itself or to the right of $g_1$, so $g\not\in A$. Similarly, $g\not\in C$. If the allocation $(A,B,C)$ is returned, $g$ belongs to the middle bundle $B$. Otherwise, we must have $g\neq g_2$, and the algorithm partitions $A'\cup (B'\backslash\{g_2\})$ into two blocks according to Lemma~\ref{lem:cutandchoose}. The subset of $A'\cup (B'\backslash\{g_2\})$ to the left of $g$ has value at least $u(G)/3$, while the subset to the right of $g$ together with $g$ has value at most $u(G)/3$. It follows that in the resulting partition, $g$ must belong to $B'$, which is the middle bundle in the allocation $(A',B',C')$.
\end{proof}

We now proceed to prove Lemma~\ref{LEM:SEPARATE}.

\begin{proof}[Proof of Lemma~\ref{LEM:SEPARATE}]
Denote by $u$ the common utility function of the agents, and assume without loss of generality that $u(g_1)\geq u(g_2)\geq u(g_3)$. We consider three cases.

\begin{itemize}
\item \emph{Case~1}: There exists a good $g\in\{g_1,g_2,g_3\}$ such that $u(g)\geq u(G)/3$. Then it must be the case that $u(g_1)\geq u(G)/3$. Set $A=\{g_1\}$, arrange the remaining goods in a line with $g_2$ and $g_3$ at the two ends, and partition the goods into two contiguous blocks according to Lemma~\ref{lem:cutandchoose}. Set $B$ and $C$ to be the two blocks and return the allocation $(A,B,C)$. Clearly, all three goods belong to different bundles. By Lemma~\ref{lem:cutandchoose}, $a_2$ and $a_3$ do not envy each other up to one good. As $a_1$ receives only one good, $a_2$ and $a_3$ do not envy her up to one good. Since $u(A)\geq u(G)/3$, we have $u(B\cup C)\leq 2u(G)/3$. By Lemma~\ref{lem:cutandchoose} again, there is a good in $a_2$'s bundle such that if we remove it, then the remaining value of $a_2$ is at most $u(B\cup C)/2\leq u(G)/3$. This implies that $a_1$ does not envy $a_2$ up to one good. A similar argument holds for $a_1$ towards $a_3$. 

\item \emph{Case~2}: There exists a good $g\not\in\{g_1,g_2,g_3\}$ such that $u(g)\geq u(G)/3$. Set $A=\{g_3,g\}$, arrange the remaining goods in a line with $g_1$ and $g_2$ at the two ends, and partition the goods into two contiguous blocks according to Lemma~\ref{lem:cutandchoose}. Set $B$ and $C$ to be the two blocks and return the allocation $(A,B,C)$. Clearly, all three goods belong to different bundles. By Lemma~\ref{lem:cutandchoose}, $a_2$ and $a_3$ do not envy each other up to one good. Moreover, since $u(g_3)\leq u(g_1),u(g_2)$, both agents do not envy $a_1$ when $g$ is removed from $a_1$'s bundle. A similar argument as in Case~1 shows that $a_1$ does not envy $a_2$ or $a_3$ up to one good.

\item \emph{Case~3}:  $u(g)< u(G)/3$ for every good $g$. Arrange the goods in a line starting with $g_1$ and $g_2$ at the left and right ends, respectively. Then, keeping $g_3$ aside, add one good at a time to the left end (to the right of $g_1$) until the total value of the goods at the left end exceeds $u(G)/3$. Since $\max(u(g_1),u(g_2),u(g_3))<u(G)/3$, this occurs when we add some good $g\not\in\{g_1,g_2,g_3\}$. Add $g_3$ to the right of $g$. If $u(L_{g_3})\geq u(G)/3$ and $u(L_{g_3}\cup\{g_3\})\leq 2u(G)/3$, add the remaining goods arbitrarily and run Algorithm~2; Lemma~\ref{lem:middle} implies that $g_3$ belongs to the middle bundle of the resulting allocation. Else, $u(\{g_3,g\})\geq u(G)/3$. Set $A=\{g_3,g\}$ and remove these two goods from the line. Add the remaining goods arbitrarily to the line, and partition the goods into two contiguous blocks according to Lemma~\ref{lem:cutandchoose}. Set $B$ and $C$ to be the two blocks and return the allocation $(A,B,C)$. A similar argument as in Case~2 shows that the resulting allocation is EF1.
\end{itemize}

Algorithm~2 uses $O(\log m$) queries, and by Lemma~\ref{lem:cutandchoose}, partitioning into two contiguous blocks according to the lemma also uses $O(\log m)$ queries. In Case~3, we can find $g$ by adding all goods except $g_3$ to the line and using binary search; this takes $O(\log m)$ queries. To determine whether there exists a good $g$ with $u(g)>u(G)/3$, arrange the goods in a line, and use binary search to find the leftmost good $g_l$ such that $u(L_{g_l}\cup\{g_l\})>u(G)/3$ and the rightmost good $g_r$ such that $u(R_{g_r}\cup\{g_r\})>u(G)/3$. Such a good $g$ must be one of $g_l$ and $g_r$. This also takes $O(\log m)$ queries. The remaining operations of the algorithm only require a constant number of queries. Hence the total number of queries is $O(\log m)$.
\end{proof}

Note that for two agents, an analogous statement holds even when the agents have arbitrary monotonic utilities, since we can place the two designated goods at different ends of a line and apply Algorithm~1.

\subsection{Arbitrary Additive Utilities}

With Algorithm~2 and Lemma~\ref{LEM:SEPARATE} in hand, we are now ready to present an algorithm that computes an EF1 allocation for three agents with arbitrary additive utilities using a logarithmic number of queries. The algorithm is based on the Selfridge-Conway procedure for computing an envy-free allocation of \emph{divisible} goods, often modeled as a cake, for three agents. At a high level, the Selfridge-Conway procedure operates by letting the first agent divide the cake into three equal pieces and letting the second agent trim her favorite piece so that it is equal to her second favorite piece. Then, the procedure allocates one ``main'' piece to each agent, with the third agent choosing first, and allocates the leftover cake in a carefully designed way. 

Like the Selfridge-Conway procedure, our algorithm starts by letting the first agent divide the goods into three almost equal bundles using Algorithm~2, so that no matter how the bundles are allocated, the agent finds the allocation to be EF1. It then proceeds by letting the second agent trim her favorite bundle so that her value for the bundle goes just below that for her second favorite bundle. However, a difficulty in our indivisible goods setting is that at this point, the second agent might find the remaining part of her favorite bundle to be worth less than her second favorite bundle \emph{even} if we remove any good from her second favorite bundle. This is possible, for instance, if the last good that she removes from her favorite bundle is of high value, and her second favorite bundle only consists of goods of low value. We will need to fix this problem by finding ``large'' goods in the leftover bundle that help us recover the EF1 guarantee; this is done in Step~3 of the algorithm. While identifying these large goods can be done easily if we can make queries for the value of every good in the leftover bundle, we would not achieve the logarithmic bound if the leftover piece contained more than a logarithmic number of goods.

\begin{framed}
\noindent
\textbf{Algorithm~3} (for three agents with additive utilities) 
\vspace{3mm}

\noindent
\emph{Step~1:} Compute an EF1 allocation $(A,B,C)$ for three identical agents with utility $u_1$. 
\begin{itemize}
\item If $a_2$ and $a_3$ have different favorite bundles among $A,B,C$, give them their favorite bundles, and give the remaining bundle to $a_1$.
\item Else, assume without loss of generality that $u_2(A)> u_2(B)\geq u_2(C)$ and $u_3(A)>\max\{u_3(B),u_3(C)\}$; the algorithm proceeds analogously in the other symmetric cases. Proceed to Step~2.
\end{itemize} 
\vspace{3mm}

\noindent
\emph{Step~2:} Let $a_2$ divide $A$ into $A'$ and $T$ such that $u_2(A')\leq u_2(B)$ and there exists $g_t\in T$ with $u_2(A'\cup\{g_t\})>u_2(B)$. 
\begin{itemize}
\item If $u_3(A')\geq \max\{u_3(B),u_3(C)\}$, give $A'$ to $a_3$, $B$ to $a_2$, and $C$ to $a_1$. Compute an EF1 allocation $(T_1,T_2,T_3)$ of the goods in $T$ for three identical agents with utility $u_2$. Let $a_3$ choose her favorite bundle followed by $a_1$, and let $a_2$ take the remaining bundle.
\item Else, we have $u_3(A')<\max\{u_3(B),u_3(C)\}$. Proceed to Step~3.
\end{itemize}
\vspace{3mm}

\noindent
\emph{Step~3:} Define $d=u_2(B)-u_2(A')\geq 0$. Call a good $g$ \emph{large} if $g\in T$ and $u_2(g)\geq d$, where we update $A'$, $T$, and $d$ during the course of the step. (Intuitively, a large good, when added to $A'$, tips $a_2$'s preference between $A'$ and $B$.) Let $a_2$ find up to three large goods. The first large good is $g_t\in T$ such that $u_2(A'\cup\{g_t\})>u_2(B)$. To find further large goods, let $E\subseteq T$ be such that $u_2(E)\geq d$, $u_2(E\backslash\{g_e\})<d$ for some $g_e\in E$, and $E$ does not contain any identified large good. 
\begin{itemize}
\item If such a set $E$ (and good $g_e$) exists, remove the goods in $E\backslash\{g_e\}$ from $T$ and add them to $A'$, and decrease $d$ by $u_2(E\backslash\{g_e\})$. (For the first large good, take $E=\{g_t\}$.) Then $g_e$ is a new large good.
\begin{itemize}
    \item If $u_3(A')\geq \max\{u_3(B),u_3(C)\}$ with the updated set $A'$, allocate the goods according to Step~2.
    \item Else, it still holds that $u_3(A')<\max\{u_3(B),u_3(C)\}$.
\end{itemize}
\item On the other hand, if no such set $E$ exists, remove all goods except the (up to two) identified large goods from $T$ and add them to $A'$, and decrease $d$ by $a_2$'s value for these goods.
\end{itemize}
\vspace{3mm}

\noindent
\emph{Step~4:} Let $S_3$ be $a_3$'s preferred bundle between $B$ and $C$, and let $S_1$ be the other bundle. Give $S_2:=A'$ to $a_2$, $S_3$ to $a_3$, and $S_1$ to $a_1$. \\

\noindent
\emph{Step~5:} Compute an EF1 allocation $(T_1,T_2,T_3)$ of the goods in $T$ for three identical agents with utility $u_3$ in such a way that all identified large goods belong to different bundles. \\

\noindent
\emph{Step~6:} Check whether there is an identified large good in each of $T_1$, $T_2$, and $T_3$.
\begin{itemize}
\item If there is an identified large good in each of $T_1$, $T_2$, and $T_3$, give $a_2$ her favorite bundle $T_i$ if we were to remove the identified large good from each of these bundles. Let $a_1$ choose her preferred bundle from the remaining two bundles (without removing the large goods), and give $a_3$ the remaining bundle.
\item Else, give the first identified large good to $a_2$ and the second identified large good (if exists) to $a_1$.
\end{itemize}
 
\end{framed}

The following theorem, which we view as our main result, establishes the claimed properties of Algorithm~3 by leveraging the machinery developed above.

\begin{theorem}
\label{thm:threeadditivealgo}
For three agents with additive utilities, Algorithm~3 computes an EF1 allocation. Moreover, the algorithm can be implemented to use $O(\log m)$ queries in the worst case.
\end{theorem}

\begin{proof}
We first show that the algorithm computes an EF1 allocation. We consider three cases. In Cases~2 and 3, assume without loss of generality that $T_1$ is the first bundle picked from among $T_1,T_2,T_3$, followed by $T_2$ and then $T_3$.

\begin{itemize}
\item \emph{Case~1}: The algorithm terminates in Step~1. Both $a_2$ and $a_3$ get their favorite bundles and therefore do not envy any other agent, while $a_1$ does not envy any other agent up to one good no matter which bundle she gets. 

\item \emph{Case~2}: The algorithm terminates in Step~2 or 3. This means that $u_3(A')\geq \max\{u_3(B),u_3(C)\}$ (either before or after finding large goods). In this case, the allocation is $(C\cup T_2,B\cup T_3,A'\cup T_1)$. Since $a_3$ gets her favorite bundles $A'$ and $T_1$, she does not envy any other agent. Next, $a_2$ gets her favorite bundle $B$, and the allocation $(T_1,T_2,T_3)$ of $T$ is computed according to her utility, so she does not envy any other agent up to one good. Furthermore, note that $a_3$'s bundle $A'\cup T_1$ is a subset of $A$, and $a_1$ would not envy $a_3$ up to one good even if $a_3$ were to get the whole bundle $A$. Also $a_1$ prefers $T_2$ to $T_3$ and the allocation $(A,B,C)$ is computed according to her utility, so she does not envy $a_2$ up to one good.

\item \emph{Case~3}: The algorithm terminates in Step~6. Denote by $T_i'$ the bundle among $T_1,T_2,T_3$ allocated to agent $a_i$. In this case, the allocation is $(S_1\cup T_1', S_2\cup T_2', S_3\cup T_3')$. Note that any identified large good always remains large. Since $S_2\cup T_2'\subseteq A$ and the allocation $(A,B,C)$ is computed according to $a_1$'s utility, $a_1$ does not envy $a_2$ up to one good. Since $a_1$ prefers $T_1'$ to $T_3'$, she also does not envy $a_3$ up to one good. Next, $a_3$ prefers $S_3$ to both $S_1$ and $S_2$, and the allocation $(T_1,T_2,T_3)$ of $T$ is computed according to her utility, so she does not envy any other agent up to one good. If there are fewer than three identified large goods, then $T$ consists of at most two (large) goods. Since $a_2$ prefers $S_2\cup T_2'$ to $B$ and $C$, and both $T_1'$ and $T_3'$ contain at most one good, $a_2$ does not envy any other agent up to one good. Else, each $T_i'$ contains an identified large good; let $g'$ be the large good in $T_2'$. We have $u_2(S_2\cup\{g'\})\geq u_2(B)\geq u_2(C)$. Moreover, $a_2$ chooses her favorite bundle from $T_1,T_2,T_3$ if we were to remove the identified large good from each bundle. Therefore she does not envy any other agent up to one good. Hence the allocation is EF1.
\end{itemize}

We now show that the algorithm can be implemented to use $O(\log m)$ queries. Step~1 can be done using Algorithm~2 with $O(\log m)$ queries. Step~2 can be done with $O(\log m)$ queries by arranging the goods in $A$ in a line and using binary search to find the leftmost good $g_t$ such that $u_2(L_{g_t}\cup\{g_t\})>u_2(B)$, and by using Algorithm~2. Finding a large good in Step~3 can be done similarly using binary search. Step~5 can be done using Lemma~\ref{LEM:SEPARATE} with $O(\log m)$ queries, and Steps~4 and 6 can be done using a constant number of queries. Hence the total number of queries is $O(\log m)$.
\end{proof}

\subsection{Arbitrary Monotonic Utilities}

What happens if the agents have arbitrary monotonic utilities? Such utilities cannot be handled by the Selfridge-Conway procedure, which is designed for the cake cutting setting where additivity is assumed. Recently, \citet{BCFI+19} presented an algorithm that produces a contiguous EF1 allocation for three agents with monotonic utilities. We show that by adapting their algorithm, which is in turn based on Stromquist's moving-knife cake cutting protocol which produces a contiguous envy-free allocation for three agents \citep{Strom80}, it is possible to compute an EF1 allocation using a polylogarithmic number of queries. The details can be found in Appendix~\ref{sec:three-app}.

\begin{theorem}
\label{thm:threemonotonic}
For three agents with arbitrary monotonic utilities, there exists an algorithm that computes an EF1 allocation using $O(\log^2 m)$ queries in the worst case.
\end{theorem}

\section{Any Number of Agents}
\label{sec:many}

In this section, we consider the general setting where there can be any number of agents. We show that even at this level of generality and with arbitrary monotonic utilities, we still only need a number of queries that is linear in both the number of agents and the number of goods in order to compute an EF1 allocation. Furthermore, we identify certain subclasses of utilities for which we can use fewer queries. We then complement our positive results by showing a lower bound on the number of queries required in this general setting.

\subsection{Upper Bounds}

Our starting point is the \emph{envy cycle elimination algorithm} of \citet{LMMS04}, which computes an EF1 allocation for agents with arbitrary monotonic utilities. The algorithm works by allocating one good at a time in arbitrary order. It also maintains an \emph{envy graph}, which has the agents as its vertices, and a directed edge from $a_i$ to $a_j$ if $a_i$ envies $a_j$ with respect to the current (partial) allocation. At each step, the next good is allocated to an agent with no incoming edge, and any cycle that arises as a result is eliminated by giving $a_j$'s bundle to $a_i$ for any edge from $a_i$ to $a_j$ in the cycle. This allows the algorithm to maintain the invariant that the envy graph has no cycles and therefore has an agent with no incoming edge before each allocation of a good. The envy cycle elimination algorithm runs in time $O(n^3m)$ in the worst case. We refer to the paper of \citet{LMMS04} for the proof of correctness and detailed analysis of this algorithm.

Our main positive result for this section is the observation that the envy cycle elimination algorithm can be implemented using a relatively modest number of (value) queries. 

\begin{theorem}
\label{thm:manyupper-envygraph}
For any number of agents with arbitrary monotonic utilities, the envy cycle elimination algorithm can be implemented to compute an EF1 allocation using:
\begin{enumerate}
\item $O(nm)$ queries in the worst case.

\item $O(n^3k\log m)$ queries in the worst case,  if the utility function of each agent takes at most $k$ (possibly unknown) values across all subsets of goods.
\end{enumerate} 
\end{theorem}

\begin{proof}
We prove the two parts in turn.
\begin{enumerate}
\item Note that in the envy cycle elimination algorithm~\citep{LMMS04}, it suffices to query the value of each agent for the $n$ bundles in each partial allocation in order to construct the envy graph. Since there are $m$ partial allocations, this takes $O(nm)$ queries. The cycle elimination step does not require additional queries because the bundles remain the same and the algorithm already knows the value of every agent for every bundle.

\item Fix an ordering of the goods to be allocated, and assume that we allocate them from left to right. Let $a_i$ be an agent with no incoming edge in the envy graph corresponding to the current (partial) allocation. Let $g$ be the leftmost unallocated good such that if we allocate all goods up to and including $g$ to $a_i$, then the value of some agent for $a_i$'s bundle increases. We then allocate all of these goods to $a_i$ at once. This is a correct implementation of the algorithm because before $g$ is allocated, the value of any agent for any bundle in the partial allocation (and thus also the envy graph) remains unchanged.

By monotonicity, we can find the good $g$ with $O(n\log m)$ queries using binary search. Since there are $n$ bundles and the value of each agent for each bundle can change up to $k-1$ times, the number of value changes is at most $n^2(k-1)$. Hence the total number of queries is $O(n^3k\log m)$.
\end{enumerate}
This completes the proof.
\end{proof}

Theorem~\ref{thm:manyupper-envygraph} illustrates a sharp contrast between EF1 and the stronger fairness notions of envy-freeness and EFX. For the latter two notions, computing a fair allocation requires an exponential number of queries in the worst case, even in the most restricted setting of two agents with identical utilities. On the other hand, for EF1 we can get away with only $O(nm)$ queries even in the most general setting of any number of agents with arbitrary monotonic utilities. Moreover, if $n$ and $k$ are small compared to $m$, the bound of Item (2) of the theorem can be better than that of Item (1). In particular, if $n$ and $k$ are constant, the implementation only requires $O(\log m)$ queries. The case $k=2$ corresponds to the setting where each agent either approves or disapproves each subset of goods.\footnote{This is not to be confused with what we call binary utilities in this paper, for which $k$ can be as large as $m$.} A small value of $k$ may occur in settings where the mechanism designer gives a predefined set of preferences that the agents can express on each subset of goods, e.g., `very interested', `somewhat interested', and `not interested'.

The next class of utilities that we consider is the class where a subset of larger size is always weakly preferred to a subset of smaller size. In other words, getting a larger bundle cannot make an agent worse off. This class of utilities applies to settings where all of the goods are roughly equally valuable and there are only minor differences within each agent's utilities for the goods. While this is a rather restricted class of utilities, the utility function used by \citet{PlautRo20} to show that computing an EFX allocation takes an exponential number of queries in the worst case belongs to this class. The following result says that if we relax the fairness notion to EF1, then it is possible to compute a fair allocation using a much smaller number of queries that does not even depend on the number of goods and is only quadratic in the number of agents.

\begin{theorem}
\label{thm:manyupper-morebetter}
For any number of agents with monotonic utilities with the property that $u_i(G_1)\leq u_i(G_2)$ for all $i$ and all $G_1,G_2\subseteq G$ with $|G_1|<|G_2|$, there exists a deterministic algorithm that computes an EF1 allocation using $O(n^2)$ queries in the worst case.
\end{theorem}

\begin{proof}
Let $q=\lfloor m/n\rfloor$ and $r=m-nq$. Divide the goods arbitrarily into $n$ bundles of $q$ goods and $r$ leftover goods. Order the agents in arbitrary order, and let each of the first $n-r$ agents choose the bundle for which she has the most value among the remaining bundles. Then, give each of the remaining $r$ agents one of the remaining bundles along with one of the $r$ leftover goods. Since we only need to know the value of the agents for at most $n$ bundles, the algorithm can be implemented using $O(n^2)$ queries.

We claim that the resulting allocation is EF1. Indeed, it follows from the assumption on the agents' utilities that no agent envies another agent with fewer or the same number of goods when a good is removed from the latter agent's bundle. Moreover, because of the agents' choices, each agent $a_i$ among the first $n-r$ agents does not envy any agent $a_j$ among the last $r$ agents when the leftover good is removed from $a_j$'s bundle. This completes the proof.
\end{proof}

Next, we consider the setting where the agents have identical monotonic utilities. It is known that an EFX allocation always exists in this setting \citep{PlautRo20}. The following lemma shows that if we relax the fairness notion to EF1, we can find a fair allocation that is moreover contiguous. Since contiguity is useful in several situations (see the remark preceding Algorithm~2), the result may also be of independent interest.

\begin{lemma}
\label{lem:manycontiguous}
Assume that the goods lie on a line. For any number of agents with identical monotonic utilities, there exists a contiguous EF1 allocation.
\end{lemma}

\begin{proof}
Suppose that the goods lie in the order $g_1,g_2,\dots,g_m$, and denote by $u$ the common utility function. For any block $G'$ of consecutive goods with $g_l$ and $g_r$ as the leftmost and rightmost goods respectively, let $w(G')=\min\{u(G'\backslash\{g_l\},G'\backslash\{g_r\})\}$. (If $G'$ is empty, set $w(G')=0$.) One can check that $w(G_1)\leq w(G_2)$ for any two blocks $G_1\subseteq G_2$. Let $S$ be the set of values of all blocks of consecutive goods (including the empty block, of which we define the value to be 0). For any value $x\in S$ and any $k\in\{1,2,\dots,n\}$, define $T_k(x)$ to be the set of all $j\in\{1,2,\dots,m\}$ for which there exist $k$ consecutive blocks $G_1,G_2,\dots,G_k$ starting from the leftmost good $g_1$ such that $w(G_i)\leq x\leq u(G_i)$ for all $1\leq i\leq k$ and the block $G_k$ ends with the good $g_j$. Our goal is to show that $m\in T_n(x)$ for some $x\in S$; this will immediately imply the desired result.

We claim that for any $x$ and $k$, $T_k(x)$ forms a (possibly empty) block of consecutive integers. To prove the claim, we fix $x$ and use induction on $k$. The base case $k=1$ follows from the observation that both $u$ and $w$ are monotonic. For the inductive step, assume that $T_k(x)=\{t,t+1,\dots,t+l\}$ for some $t$ and $l$. Then $T_{k+1}(x)$ consists of all $j$ such that the block $G'$ from $g_i$ to $g_j$ satisfies $w(G')\leq x\leq u(G')$ for some $i\in\{t+1,t+2,\dots,t+l+1\}$. Hence it suffices to show that if $j<m$ and $g_j$ is the rightmost good such that the inequalities are still satisfied when the block starts at $g_i$ and ends at $g_j$, then the inequalities are still satisfied when the block starts at $g_{i+1}$ and ends at $g_l$, for at least one $l\in\{j,j+1\}$. We consider three cases.
\begin{itemize}
\item \emph{Case 1}: $j=i$. By definition of $j$, we have $u(g_{i+1})>x$. Therefore the block consisting of only $g_{i+1}$ satisfies the inequalities. 
\item \emph{Case 2}: $j=i+1$. By definition of $j$, we have $u(\{g_{i+1}, g_{i+2}\})>x$. If $u(g_{i+1})\leq x$, the block starting at $g_{i+1}$ and ending at $g_{i+2}$ satisfies the inequalities; else, the block consisting of only $g_{i+1}$ satisfies the inequalities.
\item \emph{Case 3}: $j\geq i+2$. In this case, denote by $G'$ the (possibly empty) set $\{g_{i+2},g_{i+3},\dots,g_{j-1}\}$. By definition of $j$, we have $u(G'\cup\{g_i,g_{i+1},g_j\})>x$, $u(G'\cup\{g_{i+1},g_j,g_{j+1}\})>x$, and $\min\{u(G'\cup\{g_i,g_{i+1}\}),u(G'\cup\{g_{i+1},g_j\})\}\leq x$. If $u(G'\cup\{g_{i+1},g_j\})\geq x$, then since $u(G'\cup\{g_{i+1}\})\leq \min\{u(G'\cup\{g_i,g_{i+1}\}),u(G'\cup\{g_{i+1},g_j\})\}\leq x$, the block starting at $g_{i+1}$ and ending at $g_j$ satisfies the inequalities. Else, we have $u(G'\cup\{g_{i+1},g_j\})< x$. It follows that the block starting at $g_{i+1}$ and ending at $g_{j+1}$ satisfies the inequalities.
\end{itemize} 
This concludes the inductive step and hence the claim. It also follows that if $j<m$ is the largest element of $T_k(x)$, we can construct $k$ blocks starting at $g_1$ and ending at $g_j$ all of which satisfy the inequalities for $x$ by greedily taking each block to be the longest block such that the inequalities are satisfied.

Next, we show that for any $x$ and $k$, if $T_k(x)$ is nonempty and does not contain $m$, then the intersection of $T_k(x)$ and $T_k(y)$ is nonempty, where $y$ is the smallest element of $S$ larger than $x$. (Note that $y$ must exist since $T_k(u(G))$ is either empty or contains $m$.) We prove by induction on $k$ that the largest element of $T_k(x)$ also belongs to $T_k(y)$. For the base case $k=1$, suppose that the largest element of $T_1(x)$ is $i$. This means that $u(\{g_1,g_2,\dots,g_i\})>x$. By definition of $y$, we have $u(\{g_1,g_2,\dots,g_i\})\geq y$, which implies that $i\in T_1(y)$. For the inductive step, assume that the statement holds for $k$, and that $T_{k+1}(x)$ is nonempty and does not contain $m$. Let $j$ be the largest element of $T_{k+1}(x)$. By the remark following the claim, we know that we can greedily construct $k+1$ blocks starting at $g_1$ and ending at $g_j$ each of which satisfies the inequalities for $x$. In particular, the $k$th block will end at $g_i$, where $i$ is the largest element of $T_k(x)$. By the inductive hypothesis, $i\in T_k(y)$. A similar argument to the one used in the base case shows that $j\in T_{k+1}(y)$, as claimed.

Finally, note that $T_n(0)$ is nonempty. If $m\not\in T_n(x)$ for all $x\in S$, the previous paragraph implies that $T_n(x)$ is nonempty for all $x$. But this is impossible since $T_n(u(G))$ is either empty or contains $m$, so it must be that $m\in T_n(x)$ for at least one $x$, as desired.
\end{proof}

Like Lemma~\ref{lem:threeidenticalalgo}, Lemma~\ref{lem:manycontiguous} guarantees the existence of an EF1 allocation with the extra property that if some agent envies another agent, then the envy can be eliminated by removing one of the goods at the end of the latter agent's block.

We now leverage Lemma~\ref{lem:manycontiguous} to show that for agents with identical monotonic utilities, it is possible to compute an EF1 allocation that is moreover contiguous using a number of queries that depends only logarithmically on the number of goods. For this result we assume that the value of an agent for any subset of goods is an integer that is at most some value $K$. This is a realistic assumption for practical purposes; for instance, Spliddit lets users specify their value for each good as an integer between 0 and 1000.

\begin{theorem}
\label{thm:manyupper-identical}
Assume that the goods lie on a line. For any number of agents with identical monotonic utilities such that the utility of an agent for any subset of goods is an integer at most $K$, there exists a deterministic algorithm that computes a contiguous EF1 allocation using $O(n\log m(n\log m+\log K))$ queries in the worst case.
\end{theorem}

\begin{proof}
Using the notation from the proof of Lemma~\ref{lem:manycontiguous}, we know that $m\in T_n(x)$ for some $x\in S$. Since $S\subseteq\{0,1,2,\dots,K\}$, we may take $S$ to be this set. By monotonicity, we can use binary search to find $x$ such that $m\in T_n(x)$. For each value of $x$, we try to construct $k$ blocks that satisfy the inequality for $x$ greedily in two ways, one by taking as few goods as possible for each block, and the other by taking as many goods as possible. If the former construction does not take all goods by the $k$th block while the latter construction does not leave enough goods for the $k$th block, we know that this value of $x$ works. The existence of such an $x$ is guaranteed by Lemma~\ref{lem:manycontiguous}. Using binary search allows us to try $O(\log K)$ values of $x$, and for each value of $x$ we make $O(n\log m)$ queries. Hence the step of finding $x$ takes $O(n\log m\log K)$ queries.

Once we have $x$, it remains to construct the $n$ blocks $G_1,G_2,\dots,G_n$ that satisfy the inequalities $w(G_i)\leq x\leq u(G_i)$ for all $i$. If we have constructed $j-1$ blocks, we use binary search to find the $j$th block such that the remaining blocks can also be constructed to satisfy the inequalities. Again, we can check whether the remaining blocks can be constructed by trying to construct the blocks greedily in two ways; this takes $O(n\log m)$ queries. Since we construct $n$ blocks and we use binary search to determine each block, the total number of queries in this step is $O(n^2\log^2m)$. Combining the queries in the two steps yields the desired result.
\end{proof}

In particular, if $n$ and $K$ are constant, the bound in Theorem~\ref{thm:manyupper-identical} becomes $O(\log^2m)$. 

\subsection{Lower Bound}

To complement the positive results, we conclude by giving a lower bound (which, sadly, does not match the upper bound in Theorem~\ref{thm:manyupper-envygraph}) on the number of queries needed to compute an EF1 allocation.

\begin{theorem}
\label{thm:manylower}
Let $m\geq n^\alpha$ for some (constant) $\alpha>1$. Any deterministic algorithm that computes an EF1 allocation for $n$ agents with binary utilities uses $\Omega(n\log m)$ queries in the worst case.
\end{theorem}

\begin{proof}
Assume first that $n$ is even, say $n=2k$, and that each agent has value 1 for two goods and 0 for the remaining goods. Suppose further that for $i=1,2,\dots,k$, agents $a_{2i-1}$ and $a_{2i}$ have identical utilities; we abuse notation and denote this utility function by $u_i$. Note that if both of the goods valued by some agent are allocated to a single agent, the resulting allocation cannot be EF1.

Initially, for each $i=1,2,\dots,k$, let $G_i$ be the whole set of goods. As long as $|G_i|>2$, we answer the query of the algorithm on the value of $u_i(H)$ for a subset $H$ of goods as follows. If $|G_i\cap H|\geq |G_i|/2$, answer 2 and replace $G_i$ by $G_i\cap H$; else, answer 0 and replace $G_i$ by $G_i\backslash H$. At each step, the only information that the algorithm has is that both valued goods according to $u_i$ are contained in $G_i$. While $|G_i|>n$, any allocation contains a bundle with at least two goods from $G_i$, so the available information is not sufficient for the algorithm to return an allocation such that the two valued goods are guaranteed to be in different bundles.
This means that the algorithm must keep making queries until $|G_i|\leq n$ for every $i$. Since initially $|G_i|=m$ and the size of $G_i$ decreases by no more than half with each query, the algorithm uses at least $k\log(m/n)$ queries in the worst case. The conclusion follows from the observation that $\log(m/n)\geq\frac{\alpha-1}{\alpha}\cdot\log m$.

If $n$ is odd, we can assume that the last agent has value 0 for all goods and deduce the same asymptotic bound using the remaining $n-1$ agents.
\end{proof}

Since the assumption of Theorem~\ref{thm:manylower} holds for any constant $n$ if $m$ is large enough, and when $n=2$ the two agents considered in the proof have identical utilities and each agent values only two goods, this theorem implies Proposition~\ref{prop:EF1binary}.

\section{Conclusion}

In this paper, we investigate the query complexity of fairly allocating indivisible goods using the well-studied envy-freeness relaxations EF1 and EFX. 
On the one hand, we show that computing an EFX allocation requires a linear number of queries even for two agents with identical additive utilities.
On the other hand, we demonstrate that EF1 is much more tractable: there exist algorithms that compute an EF1 allocation for two agents with arbitrary monotonic utilities and for three agents with additive utilities using a logarithmic number of queries, as well as for three agents with arbitrary monotonic utilities using a polylogarithmic number of queries.
These results suggest that EF1 is the more appropriate fairness notion when the number of goods is large.

From a technical viewpoint, the main take-home message of our work is the following: Envy-free cake cutting protocols, designed for \emph{divisible goods}, can be adapted to yield EF1 allocations of \emph{indivisible goods} using a logarithmic number of queries. On a high level, the idea is to arrange the goods on a line, and approximately implement cut operations using binary search. We do this to obtain Theorem~\ref{thm:EF1monotonic}, by adapting the cut-and-choose protocol, and Theorem~\ref{thm:threeadditivealgo}, by building on the classic Selfridge-Conway protocol. 

However, making sure the approximation errors do not add up in a way that violates EF1 already becomes nontrivial when there are three agents, as illustrated by Algorithm~3 and Theorem~\ref{thm:threeadditivealgo}. Extending the approach even to four agents with arbitrary additive utilities, therefore, seems very challenging. A related difficulty is that the known envy-free cake cutting protocols for four or more agents are quite involved \citep{BT95,AM16,AM16b,ACFM+18}. Of course, it is possible that, in fact, there is a super-logarithmic lower bound on the query complexity in this case.

Another interesting direction for future work is to study the query complexity of fairness notions in other settings.
For instance, when agents have unequal entitlements, \citet{ChakrabortyIgSu20} extended EF1 to \emph{weighted envy-freeness up to one good (WEF1)}, and showed that this generalization can always be fulfilled for agents with additive utilities and arbitrary entitlements.
Nevertheless, in the presence of weights, even for two agents, one can no longer apply cut-and-choose as we do in Theorem~\ref{thm:EF1monotonic}, so it is unclear whether a logarithmic query complexity can still be obtained.

\section*{Acknowledgments}

This work was partially supported by the National Science Foundation under grants IIS-1350598, IIS-1714140, CCF-1525932, CCF-1733556, and CCF-1813188; by the Office of Naval Research under grants N00014-16-1-3075 and N00014-17-1-2428; by the European Research Council (ERC) under grant number 639945 (ACCORD); and by a Sloan Research Fellowship, a Guggenheim Fellowship, a Stanford Graduate Fellowship, and an NUS Start-up Grant.
A preliminary version of this work appeared in Proceedings of the 33rd AAAI Conference on Artificial Intelligence.
We thank Ayumi Igarashi, Dominik Peters, and Nipun Pitimanaaree for helpful discussions, and the anonymous reviewers for valuable comments.

\bibliographystyle{plainnat}
\bibliography{abb,arxiv}

\appendix

\section{Omitted Example}
\label{app:ex}

Assume that there are 14 goods $g_1,g_2,\dots,g_{14}$ lying on a line in this order. The common utility function $u$ is such that $u(g_1)=8$, $u(g_2)=10$, and $u(g_i)=1$ for $i=3,4,\dots,14$. In other words, the values of the goods on the line are
\[8, 10, 1, 1, 1, 1, 1, 1, 1, 1, 1, 1, 1, 1\]
in this order, where there are 12 goods of value 1.

In this example, the left cut point falls within $g_2$, while the right cut point falls between $g_4$ and $g_5$. Assume that the agents $a_1,a_2,a_3$ receive the left, middle, and right bundle respectively. If we round the left cut point to the left, then $a_1$ gets value 8 and envies $a_3$ even after removing any one good. On the other hand, if we round the left cut point to the right, then $a_1$ has value 18 and $a_2$ envies her even after removing any good. Note that even if we add the possibilities of rounding one cut point, ignoring the other cut point, and then dividing the remaining goods into two parts according to Lemma~\ref{lem:cutandchoose}, the same example still shows that none of these additional possibilities works. On the other hand, Algorithm~2 returns the allocation $(\{g_1\},\{g_2,\dots,g_6\},\{g_7,\dots,g_{14}\})$, which is EF1. 

\section{Algorithm for Three Agents with Monotonic Utilities}
\label{sec:three-app}

We begin by introducing some definitions, mostly following \citet{BCFI+19} but with a few deviations. Assume that the goods lie on a line in the order $g_1,g_2,\dots,g_m$. For $j\leq k$, we use $P(g_j,g_k)$ to refer to the subsequence $g_j,g_{j+1},\dots,g_k$ as well as the set $\{g_j,g_{j+1},\dots,g_k\}$. Given a subsequence $P(g_j,g_k)$ and an agent $a_i$, the leftmost good $g$ such that $u_i(L_g\cup\{g\})\geq u_i(R_g)$ is called the \emph{lumpy tie} over $P(g_j,g_k)$, where $L_g$ and $R_g$ are defined with respect to the subsequence. One can check that if $g$ is the lumpy tie over $P(g_j,g_k)$, then $u_i(R_g\cup\{g\})\geq u_i(L_g)$. 

Next, given a subsequence $P(g_j,g_k)$, we define the \emph{median lumpy tie} to be the median of the lumpy ties of the three agents. An agent is called a \emph{left agent}, \emph{middle agent}, and \emph{right agent} if the lumpy tie of the agent is to the left of, equal to, and to the right of the median lumpy tie, respectively. By definition, there is at most one left agent, at least one middle agent, and at most one right agent. \citet{BCFI+19} showed that for any two of the three agents, it is possible to find an EF1 allocation of the goods in $P(g_j,g_k)$ to the two agents such that both agents receive a bundle that they value at least as much as both $L_g$ and $R_g$, where $g$ is the median lumpy tie. We call such an allocation a \emph{lumpy allocation} to the two agents.

Before we proceed to the description of our algorithm, we give a high-level overview of Bil\`{o} et al.'s algorithm and the changes we need to make in order to obtain the desired query complexity. Bil\`{o} et al.'s algorithm works by maintaining a left block, a middle block, and a right block. Initially the left block consists only of the leftmost good. It subsequently gets incremented by one good at a time until it is ``large enough'' to produce an EF1 allocation in combination with the other two blocks. For each left block, the middle and right blocks are computed in such a way that they are similar to each other with respect to the utilities of the three agents. When the left block gets incremented, the middle and right blocks are recomputed. However, instead of jumping from the middle and right blocks of the previous left block to those of the current left block immediately, the algorithm gradually changes the two blocks by one good at a time to ensure that no potentially valid allocation is skipped over.

To achieve low query complexity, we need to make two important changes. The first change is that instead of incrementing the left block by one good at a time, we use binary search on the left block. However, a difficulty here is that it is not clear whether the property of a left block being ``large enough'' to produce an EF1 allocation is monotone. In other words, it could be that a certain left block produces an EF1 allocation, but the left block obtained by adding one more good does not satisfy this property. We get around this difficulty by only requiring the algorithm to find a left block that produces an EF1 allocation and such that the left block obtained by \emph{removing} one good fails to do so; a left block with this property can be found using binary search. While the resulting algorithm could conceivably produce a different allocation from that of Bil\`{o} et al., it turns out that their proof of correctness can still be used for our algorithm. 

The second difficulty in using Bil\`{o} et al.'s algorithm for our purposes lies in the step where the left block gets incremented. Since the algorithm gradually changes the middle and right blocks, this step could take linear time with a naive implementation. However, we can circumvent this issue by observing that throughout this process, the left block, the left end of the middle block, and the right end of the right block all stay fixed; only the right end of the middle block and the left end of the right block change. Moreover, the agents only need to compare the left block with each of the middle and right block. As a result, we can use binary search to determine exactly when each agent finds the left block to be more valuable than the middle and right blocks.

We are now ready to describe our algorithm. Assume without loss of generality that $m\geq 2$. We sometimes abuse notation and write $u_a$ to denote the utility function of agent $a$. The algorithm maintains three bundles $L$, $M$, and $R$, each of which is a subsequence of consecutive goods on the line. Bundle $L$ always covers the left end of the line and bundle $R$ always covers the right end of the line, but together the three bundles do not contain all of the goods. Agent $a_i$ is called a \emph{shouter} if $u_i(L)\geq u_i(M)$ and $u_i(L)\geq u_i(R)$.

\afterpage{
\begin{framed}
\noindent
\textbf{Algorithm~4} (for three agents with monotonic utilities) \\

\noindent
The four steps below take an index $0\leq l\leq m-2$. Find an index $l$ such that an allocation is returned for index $l$ after these steps, and either (i) $l=0$, or (ii) no allocation is returned for index $l-1$. Return the allocation for index~$l$. \\

\noindent
\emph{Step~1:}  Set $L=P(g_1,g_{l+1})$, $M=P(g_{l+2},g_{r-1})$, and $R=P(g_{r+1},g_m)$, where $g_r$ is the median lumpy tie over $P(g_{l+2},g_m)$.  \\

\noindent
\emph{Step~2:} If some agent $a_{\text{left}}$ shouts, $a_{\text{left}}$ gets the left bundle $L$. Allocate the remaining goods according to the lumpy allocation for the remaining two agents.  \\

\noindent
\emph{Step~3:} Set $M=P(g_{l+3},g_{r-1})$. If at least two agents shout, we claim that there is a shouter $a$ who is a middle agent over $P(g_{l+2},g_m)$; this claim is established in the proof of Theorem~\ref{thm:threemonotonicalgo}. Allocate $L$ to a shouter $a_{\text{left}}$ distinct from $a$. Let the agent $a'$ distinct from $a$ and $a_{\text{left}}$ choose her preferred bundle between $P(g_{l+2},g_{r-1})$ and $P(g_r,g_m)$. Agent $a$ receives the other bundle. \\

\noindent
\emph{Step~4:} If $g_r$ is the median lumpy tie over $P(g_{l+3},g_m)$, directly move to  Steps 4(a)--4(c). If $g_r$ is not the median lumpy tie over $P(g_{l+3},g_m)$, increment $r$ by~1, set $M=P(g_{l+3},g_{r-1})$ and $R=P(g_{r+1},g_m)$, and move to Steps 4(a)--4(c).  \\

\emph{Step~4(a):} If at least two agents shout, let $a$ be a shouter who did not shout in the previous step (which could be Step~3 or Step~4). If there is a shouter $a_{\text{left}}$ who shouted in the previous step, $a_{\text{left}}$ receives $L$; else, give $L$ to an arbitrary shouter $a_{\text{left}}$ different from $a$. The agent $a'$ distinct from $a$ and $a_{\text{left}}$ chooses her preferred bundle between $P(g_{l+2},g_{r-1})$ and $P(g_r,g_m)$, breaking ties in favor of the first option. Agent $a$ receives the other bundle. \\

\emph{Step~4(b):} If $g_r$ is not the median lumpy tie over $P(g_{l+3},g_m)$, repeat Step~4. \\

\emph{Step~4(c):} If $g_r$ is the median lumpy tie over $P(g_{l+3},g_m)$ and only one agent $a_{\text{left}}$ shouts, give $P(g_1,g_{l+2})$ to $a_{\text{left}}$, and allocate the remaining goods according to the lumpy allocation for the remaining two agents. Else, $g_r$ is the median lumpy tie over $P(g_{l+3},g_m)$ but no agent shouts. In this case, no allocation is returned for index $l$.  \\
\end{framed}
}

\begin{theorem}
\label{thm:threemonotonicalgo}
For three agents with arbitrary monotonic utilities, Algorithm~4 computes an EF1 allocation. Moreover, the algorithm can be implemented to use $O(\log^2 m)$ queries in the worst case.
\end{theorem}

\begin{proof}
First, we show that the algorithm is well-defined and terminates with an allocation. When $l=m-2$, both bundles $M$ and $R$ are empty in Step~1, so all agents shout in Step~2 and an allocation is returned. Therefore there exists an index $0\leq i\leq m-2$ such that an allocation is returned for index $l$, and either $l=0$ or no allocation is returned for index $l-1$. To prove that the algorithm is well-defined, we need to show that if at least two agents shout in Step~3, there is a shouter who is a middle agent over $P(g_{l+2},g_m)$. We show that no shouter can be a right agent; this suffices because there can be at most one left agent, so at least one shouter is a middle agent. Assume for contradiction that a shouter $a_i$ is a right agent. This means that $u_i(R)\geq u_i(M\cup \{g_{l+2},g_r\})$. Since $a_i$ is a shouter, we have $u_i(L)\geq u_i(R)\geq u_i(M\cup \{g_{l+2},g_r\})$. However, $a_i$ did not shout in Step~2 (when no agent shouted), so either $u_i(R)>u_i(L)$ or $u_i(M\cup\{g_{l+2}\})>u_i(L)$, both of which lead to a contradiction.

Next, we show that the allocation returned by the algorithm is EF1. If the allocation is returned for the index $l=0$, the algorithm leading up to the allocation is executed in the same way as that of \citet{BCFI+19}, so their proof also carries over. We therefore assume that the allocation is returned for some index $1\leq l\leq m-2$, and no allocation is returned for index $l-1$. We divide into cases according to the step in which the allocation is returned.

\begin{itemize}
    \item \emph{Case 1:} The allocation is returned in Step~2. Since agent $a_{\text{left}}$ is a shouter, she does not envy the other two agents up to good $g_r$. The remaining two agents do not envy each other up to one good by definition of lumpy allocation. An agent $a_i$ who is not a shouter does not envy $a_{\text{left}}$ because she prefers either $M$ or $R$ to $L$, and therefore receives a bundle preferred to $L$. Now, consider an agent $a_i\neq a_{\text{left}}$ who is a shouter. Since no allocation is returned for index $l-1$, no agent shouts in Step~4(c) for that index. Note that the bundles $M$ and $R$ in Step~4(c) for index $l-1$ are the same as those in Step~2 for index $l$, while the bundles $L$ in the two steps only differ by good $g_{l+1}$. This means that either $u_i(M)>u_i(L\backslash\{g_{l+1}\})$ or $u_i(R)>u_i(L\backslash\{g_{l+1}\})$. Since $a_i$ gets a bundle at least as good as $M$ or $R$, she does not envy $a_{\text{left}}$ up to $g_{l+1}$.
    
    \item \emph{Case 2:} The allocation is returned in Step~3. Since agent $a_{\text{left}}$ is a shouter, she does not envy the agent who receives $P(g_{l+2},g_{r-1})$ up to good $g_{l+2}$, and does not envy the agent who receives $P(g_r,g_m)$ up to good $g_r$. Agent $a'$ receives her preferred bundle between $P(g_{l+2},g_{r-1})$ and $P(g_r,g_m)$, so she does not envy the agent who receives the other bundle. Moreover, since $a'$ did not shout in Step~2, she prefers either $P(g_{l+2},g_{r-1})$ or $P(g_{r+1},g_m)$ to $L$. Since $L$ remains the same, she also prefers her chosen bundle to $L$.
    
    Consider agent $a$. Since $a$ is a middle agent, her lumpy tie over $P(g_{l+2},g_m)$ is $g_r$, meaning that $u_a(P(g_r,g_m))\geq u_a(P(g_{l+2},g_{r-1}))$. Recall that $a$ did not shout in Step~2 but does shout in Step~3. The only change between the two steps is that good $g_{l+2}$ is removed from bundle $M$. This means that $u_a(P(g_{l+2},g_{r-1}))>u_a(L)$, and so $u_a(P(g_r,g_m)) >u_a(L)$. Since $a$ receives either $P(g_{l+2},g_{r-1})$ or $P(g_r,g_m)$, she does not envy $a_{\text{left}}$, who receives $L$. Moreover, if $a'$ chooses $P(g_{l+2},g_{r-1})$, then $a$ does not envy $a'$ because $u_a(P(g_r,g_m))\geq u_a(P(g_{l+2},g_{r-1}))$. Else, $a'$ chooses $P(g_r,g_m)$. In this case, $u_a(L)\geq u_a(P(g_{r+1},g_m))$ since $a$ is a shouter. It follows that $u_a(P(g_{l+2},g_{r-1}))>u_a(L)\geq u_a(P(g_{r+1},g_m))$, and therefore $a$ does not envy $a'$ up to good $g_r$.
    
    \item \emph{Case 3:} The allocation is returned in Step~4(a). First, we claim that if $a_i$ is a shouter who did not shout in the previous step (which was either Step~3 or Step~4), then $u_i(P(g_r,g_m)) >u_i(L)\geq u_i(P(g_{l+3},g_{r-1}))$. In the previous step, the bundle $M$ was $P(g_{l+3},g_{r-2})$ and the bundle $R$ was $P(g_r,g_m)$. Since $a_i$ did not shout with these bundles, it must be that either $u_i(P(g_{l+3},g_{r-2})) > u_i(L)$ or $u_i(P(g_r,g_m)) > u_i(L)$. Since $a_i$ is a shouter, $u_i(L)\geq u_i(P(g_{l+3},g_{r-1}))$, so the first case is impossible. This means that $u_i(P(g_r,g_m)) > u_i(L)\geq u_i(P(g_{l+3},g_{r-1}))$, as claimed.
    
    We now consider the three agents in turn. Since agent $a_{\text{left}}$ is a shouter, she does not envy the agent who receives $P(g_{l+2},g_{r-1})$ up to good $g_{l+2}$, and does not envy the agent who receives $P(g_r,g_m)$ up to good $g_r$. Agent $a'$ receives her preferred bundle between $P(g_{l+2},g_{r-1})$ and $P(g_r,g_m)$, so she does not envy the agent who receives the other bundle. If $a'$ is not a shouter, she prefers either $P(g_{l+2},g_{r-1})$ and $P(g_r,g_m)$ to $L$, and therefore prefers her chosen bundle to $L$. Else, $a'$ is a shouter, meaning that all three agents are shouters, and $a'$ was not a shouter in the previous step (since there was at most one shouter). By our claim, $a'$ prefers $P(g_r,g_m)$ to $L$. Since $a'$ has a choice between $P(g_{l+2},g_{r-1})$ and $P(g_r,g_m)$, she does not envy $a_{\text{left}}$.
    
    Finally, we show that agent $a$ does not envy the other two agents up to one good. We consider two cases.
    \begin{itemize}
        \item Agent $a'$ strictly prefers $P(g_r,g_m)$ to $P(g_{l+2},g_{r-1})$, and consequently chooses $P(g_r,g_m)$. So $a$ gets $P(g_{l+2},g_{r-1})$. By definition of lumpy tie, the lumpy tie of $a'$ over $P(g_{l+2},g_m)$ is either $g_r$ or to its right. Since we started with $g_r$ being the median lumpy tie over $P(g_{l+2},g_m)$ in Step~1 and incremented $r$ at least once in Step~4, $g_r$ is strictly to the right of the median lumpy tie over $P(g_{l+2},g_m)$. This means that $a'$ is a right agent over $P(g_{l+2},g_m)$. Since there can be at most one right agent, $a$ is either a left agent or a middle agent over $P(g_{l+2},g_m)$. It follows that $u_a(P(g_{l+2},g_{r-1}))\geq u_a(P(g_r,g_m))$, and so $a$ does not envy $a'$. Moreover, since $a$ did not shout in the previous step, our claim at the beginning of Case~3 implies that $u_a(P(g_r,g_m))> u_a(L)$. Hence $u_a(P(g_{l+2},g_{r-1})) > u_a(L)$, implying that $a$ does not envy $a_{\text{left}}$.
        \item Agent $a'$ weakly prefers $P(g_{l+2},g_{r-1})$ to $P(g_r,g_m)$, and therefore chooses $P(g_{l+2},g_{r-1})$. So $a$ gets $P(g_r,g_m)$. Since $a$ did not shout in the previous step, our claim at the beginning of Case~3 implies that $u_a(P(g_r,g_m)) >u_a(L)\geq u_a(P(g_{l+3},g_{r-1}))$. It follows that $a$ does not envy $a_{\text{left}}$, and also does not envy $a'$ up to good $g_{l+2}$.
    \end{itemize}
    
    \item \emph{Case 4:} The allocation is returned in Step~4(c). Since agent $a_{\text{left}}$ is a shouter, she does not envy the other two agents up to good $g_r$. The remaining two agents do not envy each other up to one good by definition of lumpy allocation. Any agent other than $a_{\text{left}}$ is not a shouter, and so prefers either $P(g_{l+3},g_{r-1})$ or $P(g_{r+1},g_m)$ to $L$. Since such an agent receives a bundle that she prefers to both $P(g_{l+3},g_{r-1})$ and $P(g_{r+1},g_m)$, she also prefers her bundle to $L$, meaning that she does not envy $a_{\text{left}}$ up to good $g_{l+2}$.
    
\end{itemize}

We now show that the algorithm can be implemented to use $O(\log^2 m)$ queries. The lumpy tie of each agent over a subsequence of goods can be found by binary search using $O(\log m)$ queries, and therefore the median lumpy tie can also be found using $O(\log m)$ queries. Since the lumpy allocation can be found based on the lumpy ties of the three agents and a constant number of additional queries \citep{BCFI+19}, computing it also takes $O(\log m)$ queries. Determining whether an agent shouts takes a constant number of queries. Based on these observations, we claim that for each value of $l$, the four steps of the algorithm can be done using $O(\log m)$ queries. The only step for which this is not clear is Step~4, since it involves repeatedly incrementing $r$ by 1. However, note that to implement Step~4, it suffices to determine for each agent $a_i$ the largest index $j$ such that $u_i(L)\geq u_i(P(g_{l+3},g_j))$ and the smallest index $k$ such that $u_i(L)\geq u_i(P(g_k,g_m))$. These indices can be found by binary search using $O(\log m)$ queries. Finally, again using binary search, it suffices to try $O(\log m)$ indices $l$. Hence the total number of queries is $O(\log^2 m)$.
\end{proof}

\end{document}